\documentclass[superscriptaddress,aps,pra,twocolumn,showpacs,nofootinbib,longbibliography,showkeys]{revtex4-2}
\usepackage{slashbox}
\usepackage{enumitem}
\usepackage{amsmath,amssymb,amsthm}
\usepackage{easybmat,comment}
\usepackage[colorlinks=true,citecolor=blue,urlcolor=blue, linkcolor = magenta]{hyperref}
\usepackage[pdftex]{graphicx}
\usepackage{times,txfonts}
\usepackage{braket}
\usepackage{color}
\usepackage{natbib}
\usepackage{subcaption}
\usepackage{ragged2e}
\DeclareCaptionJustification{justified}{\justifying}
 \usepackage[compat=0.4]{yquant}
\newcommand{\be}{\begin{equation}}
	\newcommand{\ee}{\end{equation}}
\newcommand{\ba}{\begin{eqnarray}}
	\newcommand{\ea}{\end{eqnarray}}

\newtheorem{theorem}{Theorem}

\usepackage{comment}
\begin{document}
	
\title{Concatenating quantum error-correcting codes with decoherence-free subspaces and vice versa}
\author{Nihar Ranjan Dash\textsuperscript{}}
   \email{dash.1@iitj.ac.in}
   \affiliation{Indian Institute of Technology Jodhpur, Rajasthan 342030, India}
\author{Sanjoy Dutta}
    \affiliation{Poornaprajna Institute of Scientific Research, Bidalur Post, Devanahalli, Bengaluru 562164, India}
\author{R. Srikanth\textsuperscript{}}
   \email{srik@ppisr.res.in}
   \affiliation{Poornaprajna Institute of Scientific Research, Bidalur Post, Devanahalli, Bengaluru 562164, India}
\author{Subhashish Banerjee}
   \email{subhashish@iitj.ac.in}
   \affiliation{Indian Institute of Technology Jodhpur, Rajasthan 342030, India}




\begin{abstract}
Quantum error-correcting codes (QECCs) and decoherence-free subspace (DFS) codes provide active and passive means, respectively, to address certain types of errors that arise during quantum computation. The latter technique is suitable to correct correlated errors with certain symmetries and the former to correct independent errors. The concatenation of a QECC and a DFS code results in a degenerate code that splits into actively and passively correcting parts, with the degeneracy impacting either part, leading to degenerate errors as well as degenerate stabilizer operators. The concatenation of the two types of code can aid universal fault-tolerant quantum computation when a mix of correlated and independent errors is encountered. In particular, we show that for sufficiently strongly correlated errors, the concatenation with the DFS as the inner code provides better entanglement fidelity, whereas for sufficiently independent errors, the concatenation with the QECC as the inner code is preferable. As illustrative examples, we examine in detail the concatenation of a two-qubit DFS code and a three-qubit repetition code or five-qubit Knill-Laflamme code, under independent and correlated errors.
\end{abstract}
\maketitle


\section{Introduction}\label{sec-intro}
The loss of quantum coherence, also known as the decoherence process, is the biggest obstacle to realizing practical quantum computation and communication \cite{banerjee2018open}. Errors also arise because of imperfections and faults in the quantum communication system or computation circuit. These errors, due to decoherence and device imperfections, reduce the fidelity and reliability of quantum operations \cite{shor1995scheme}. To overcome these obstacles, a variety of methods have been employed, among them quantum error-correcting codes (QECCs) \cite{steane1996error, calderbank1996good, knill1997theory,knill2000theory}, decoherence-free subspace (DFS) \cite{zanardi1997noiseless,lidar1998decoherence,bacon2000universal}, dynamic decoupling (DD) \cite{viola1999dynamical}, and error mitigation \cite{temme2017error,li2017efficient}. Quantum error-correcting codes offer an active method of intervention to protect quantum information from a wide variety of errors, including, most generally, independent errors, by correcting them when they occur. These codes can also be used to characterize quantum dynamics \cite{omkar2015characterization, omkar2015quantum}. Decoherence-free subspace is a passive scheme to combat correlated errors and preserves the coherence of the quantum state by encoding information in subspaces that are immune to certain types of noise by virtue of symmetries in the dynamics. An interesting example of DFS arises in the study of quantum memory for photons, where a collective reservoir interaction can happen for the cluster of atoms that make up the memory. This motivates the construction of a two-dimensional decoherence-free quantum memory protected from collective errors \cite{cerf2007quantum}. Note that this is an extreme case, and in the other extreme, the atoms may be subjected to independent errors.

The practice of concatenating block codes is extensively employed in the field of quantum information science and serves as a crucial element in nearly all fault-tolerant strategies \cite{gottesman2010introduction, poulin2006optimal}. The basic concatenated code, even in the classical context \cite{forney1966concatenated}, comprises an outer and an inner code. More generally, concatenation of codes is accomplished by iterative encoding of blocks of qubits in different levels. Embedding an inner code in an outer code decreases the effective error probability of the concatenated code, making its data qubits more reliable \cite{gaitan2008quantum}. The establishment of the accuracy threshold theorem \cite{aharonov1999fault, knill1998resilient,kitaev1997quantum} relies substantially on these concatenated codes. There can be multiple layers in the concatenated quantum code \cite{knill1996concatenated, rahn2002exact, grassl2009generalized}, where the physical space of one code behaves as the logical space for the next code. One can construct an efficient and robust error-correction scheme using this hierarchical structure.

While recursive concatenation of the same QECC or different QECCs has been extensively studied \cite{rahn2002exact,fern2008correctable,grassl2009generalized,huang2015generality}, recently a few authors have explored hybrid concatenation schemes, such as DD and QECC \cite{Khodjasteh2003quantum, Paz2013optimally}, DD and DFS \cite{zhang2004concatenating,west2012exchange}, and QECC and DFS \cite{lidar1999concatenating}. The nature of errors in the given quantum information processor, for example, whether the noise is ``bursty" or independent, will determine the codes to concatenate as well the order in which they may be concatenated \cite{chamberland2017error}. Such a concatenated setup, together with the availability of transversal gate operations \cite{Jochym-O'Connor2014using, chamberland2016thresholds}, enables universal fault-tolerant quantum computation. In particular, the combination of QECC and DFS can facilitate fault-tolerance \cite{boulant2005experimental} in the presence of correlated errors \cite{clemens2004quantum, cafaro2011concatenation, cafaro2011quantifying}.
The experimental implementation of the hybrid concatenation of active and passive methods for fighting errors \cite{boulant2005experimental} paves the way for practical quantum computation in the presence of coherent errors \cite{ouyang2021avoidng}.
A specific category of DFS-QECC hybrids has been developed to effectively address spontaneous emission errors and collective dephasing specifically intended to be compatible with the quantum optical and topological implementation of quantum dots in cavities and trapped ions \cite{alber2001stabilizing, khodjasteh2002universal}. The efficiency of concatenated DFS-QECC hybrids is contingent upon the use of a practical set of universal quantum gates that can maintain states exclusively within the DFS. It also depends on the fault-tolerant execution of DFS state preparation and decoding processes \cite{bacon2000universal}.

The present work focuses on the two-layered (and, in general, multi-layered)  protection based on the concatenation of QECC and DFS, in particular comparatively studying the concatenation of DFS with QECC layer and vice versa. The need to study such a hybrid concatenation scheme naturally arises when both correlated and independent errors occur simultaneously in the setup of a quantum computer. A basic situation of this sort is depicted in Fig. \ref{fig:grid-structure}. Here each row of qubits may represent, for example, ions in a linear ion trap, which are subject to correlated errors, for which a DFS is suitable. The full quantum information processor consists of an array of such rows such that the ions along the column are coupled with independent reservoirs and hence subject to independent errors. In such a case as this, a QECC may be suitable. For fault-tolerant computation, we would require a concatenation of QECC and DFS. The question then hinges on the order of concatenating the two: whether QECC with DFS or vice versa. It is the issue that this work studies, to determine situations where one scheme or the other may be more advantageous according to different criteria.

The remaining article is organized as follows. In Sec. \ref{sec-prelim} we present preliminaries on the DFSs, QECCs, and the construction of concatenated codes. In Sec. \ref{sec:2scheme} we study in detail the concatenation of two or more codes, in both independent and correlated error models. Specifically, we discuss two schemes of concatenating a QECC and DFS, one with the latter as the inner layer and the former as the outer layer (QD code) and vice versa (DQ code). Specific examples of  concatenated codes, six-qubit codes concatenating a $[[3,1]]$ repetition QECC and $[[2,1]]$ DFS, and ten-qubit codes concatenating a $[[5,1]]$ Knill-Laflamme QECC and $[[2,1]]$ DFS, are studied in Secs. \ref{sec-conccodes6} and \ref{sec-conccodes10}, respectively. Passivity and degeneracy of the DFS are reflected as the passive part of the concatenated code and the degeneracy in the passive and active parts of the concatenated code, respectively. In particular, the degeneracy has a twofold manifestation: (a) as degenerate equivalence classes of correctable errors (Secs. \ref{sec-conccodes6} and \ref{sec-conccodes10}, covering the $[[6,1]]$ and $[[10,1]]$ codes) and (b) as the corresponding stabilizers (Sec. \ref{sec-stabilizergens}) of the concatenated code. We summarize and discuss our results in Sec. \ref{sec-conclusions}.

\begin{figure}[htp]
    \centering
    \includegraphics[width=8cm]{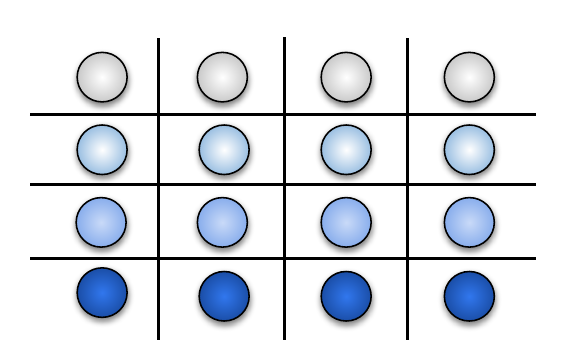}
    \caption{Independent-correlated hybrid error model. Correlated and independent noise are shown in an array of qubits constituting a quantum computer such as an optical lattice or trapped ions in a cavity QED setup. The qubits along a row are subject to correlated noise, described by a dependent error model. Qubits along a column are subject to an independent-error model.}
    \label{fig:grid-structure}
\end{figure}

\section{Preliminaries}\label{sec-prelim}
\subsection{Decoherence-free subspace}
Decoherence-free subspaces are the subspaces that act as quiet corners in the total Hilbert space shielded from errors by a certain symmetry in the system's interaction with the environment \cite{lidar2001decoherenceI,lidar2001decoherenceII,pathak2013elements}. By encoding quantum information in these subspaces, the coherence of the quantum system can be preserved for an extended period. The concatenation of DFS and QECC gives an additional layer of protection from hybrid errors, which have elements of symmetry and independence.

The Liouville equation is given by
\begin{align}
\mathcal{L}[\rho]= \Dot{\rho},
\label{eq:liouvilleeq}
\end{align}
where $\mathcal{L}$ is the Liouvillian.
A DFS constitutes the degenerate subspace obtained as a stationary solution to Eq. (\ref{eq:liouvilleeq}). States in this subspace are called decoherence-free (DF) states.

For simplicity, we consider a two-qubit system affected by the collective bit-flip error. The error group is \{$II$, $XX$\}. Since this group is Abelian, all irreducible representations belonging to it are one dimensional \cite{lidar2001decoherenceI}. It is a result in group theory that the number of irreducible representations within a group is equivalent to the number of classes associated with that group \cite{cornwell1997group}. 

As there are two classes in this group, $II$ and $XX$, the number of irreducible representations of this group is equal to 2, namely, $\Gamma^{+} \equiv \{1,1\}$ and $\Gamma^{-}\equiv \{1, -1\}$. The character is obtained by calculating the trace of the irreducible representation. Since the square of the two elements in the group is an identity, the character of the irreducible representations can only take $\pm 1$. The DFS can be constructed with the action of a projection operator belonging to the one-dimensional irrep on the initial states. Here there are two irreducible representations, which correspond to two DFSs, given by
\begin{align}
\Gamma^+ &\equiv \textrm{span}\{\ket{00}+\ket{11}, \ket{01}+\ket{10}\},\nonumber\\
 \Gamma^- &\equiv \textrm{span}\{\ket{00}-\ket{11} ,\ket{01}-\ket{10}\}.
 \label{eq:irrep1and2}
\end{align}
The two corresponding operators which project to the given DFS are
\begin{align}
P_{\pm} = \frac{1}{2} (II \pm XX).
\label{eq:projectorsoneandtwo}
\end{align}
For example, we can define our logical DF code states for $\Gamma^+$ as (apart from a normalization factor)
\begin{align}
\ket{0}_D &= \ket{00}+\ket{11},\nonumber\\
\ket{1}_D &= \ket{01}+\ket{10}.
\label{eq:dfslogic}
\end{align}
The DFS spanned by these states is denoted by $\mathcal{H}_\textsubscript{DFS}$, and it is orthogonal and complementary to the subspace $\mathcal{H}_\textsubscript{DFS}^\perp$, spanned by the DFS states defined by $\Gamma^-$, namely,
\begin{align}
\ket{2}_D &= \ket{00}-\ket{11},\nonumber\\\ket{3}_D &= \ket{01}-\ket{10}.
\label{eq:orthogonaldfslogic}
\end{align}
The two DFSs are
\begin{align}
\mathcal{H}_\textsubscript{DFS} \equiv \textrm{span}\{\ket{0}_D, \ket{1}_D\},\nonumber\\ \mathcal{H}_\textsubscript{DFS}^\perp \equiv \textrm{span}\{\ket{2}_D, \ket{3}_D\}.
\label{eq:dfsandorthogonaldfshilbertspace}
\end{align}
Each of these subspaces in Eq. \eqref{eq:dfsandorthogonaldfshilbertspace} serves as a bona fide DFS.
The superposition of DF states from the same irreducible representation is also a DF state, that is,
\begin{align}
\ket{\psi}_D &= \alpha\ket{0}_D+\beta\ket{1}_D,
\label{eq:psilogic}
\end{align}
with $|\alpha|^2 + |\beta|^2=1$, is a DF state. However, one can not take states from different irreducible representations to make a DF state.

It may be noted that the projector of an irreducible representation acts as an identity operator for the DF states corresponding to that irreducible representation, while those DF states are annihilated by the projectors corresponding to a different irreducible representation.

\subsection{Quantum Error-Correcting Codes}
The QECCs constitute a major approach for combating errors \cite{lidar2013quantum}. In the quantum error-correction scheme, quantum information is encoded in a larger Hilbert space by adding redundancy, such that errors only affect the redundancy. The error-correction process involves the encoding, detection, and recovery process. Generally, a QECC is a set of orthogonal states or quantum code words represented by the $[[n,k,d]]$ code, where $k$ logical qubits are encoded into $n$ physical qubits with the code distance $d$.

Considering the code words $\ket{\psi_{i}}$ for the code $\mathcal{C}$, the necessary and sufficient conditions for the quantum error-correcting codes are \cite{knill1997theory,bennett1996mixed}
\begin{align}
\bra{\psi_{i}}E_{m}^{\dag}E_{n}\ket{\psi_{j}}&=0 ~~(i \ne j), \nonumber\\
\bra{\psi_{i}}E_{m}^{\dag}E_{n}\ket{\psi_{i}}&= \bra{\psi_{j}}E_{m}^{\dag}E_{n}\ket{\psi_{j}}
\label{errorcorrectingcriterion2}
    \end{align}
for $\bra{\psi_{i}}\ket{\psi_{j}}=0$. Here $E_{m}$ and $E_{n}$ are the error operator elements from the correctable errors. Instead of the code words, the QECC can be represented by the minimal representation with the help of stabilizer generators. 

The set of stabilizer generators $S_{i}$ is a minimal set that can generate all the elements in the stabilizer $\mathcal{S}$. The action of the stabilizer generators $S_{i}$ on the code words is written as
\begin{align}
    S_{i}\ket{\psi_{j}}=\ket{\psi_{j}}.
\end{align}
The normalizer $\mathcal{N}$ of the stabilizer $\mathcal{S}$ is a subset of the set of Pauli operators $\mathcal{P}$ such that the stabilizer is closed under conjugation by the normalizer elements,
\begin{align}
    NSN^{\dag}\in \mathcal{S},
    \label{eq:normalizer}
\end{align}
with $N\in\mathcal{N}\subset \mathcal{P}$. 
The centralizer $\mathbb{C}$ of the stabilizer $\mathcal{S}$ is the set of Pauli operators $\mathcal{P}$ that commute with stabilizer elements
\begin{align}
    CS&=SC,
\end{align}
i.e., $\mathbb{C}=\{C \colon C\in \mathcal{P}, \forall S \in \mathcal{S}, [S,C]=0\}.$
Owing to 
the fact that $-I \not\in \mathcal{S}$, the centralizer and normalizer of the stabilizer are identical, $\mathbb{C} = \mathcal{N}$. The physical significance of the normalizer is that $\mathcal{N} - \mathcal{S}$ corresponds to the set of (nontrivial) logical operators of the code. Detectable errors $E$ are characterized by the condition $E \not\in \mathcal{N} - \mathcal{S}$. A set of errors $\{E_m\}$ is correctable if and only if $E_m^{\dagger} E_n \not\in \mathcal{N} - \mathcal{S}$.

\subsection{Construction of Concatenated Codes}\label{sec:CCC}
For a simple exposition of the concatenated quantum error-correcting code construction, we restrict the discussion at first to two levels of concatenation. Consider two quantum codes (two QECCs or two DFSs or one of each) for overcoming errors: $\mathcal{C}_{o}$, an outer $[[n_o,k_o]]$ code, which encodes $k_o$ logical qubits into $n_o$ physical qubits, and $\mathcal{C}_{i}$, an inner $[[n_i,k_i]]$ code, which encodes $k_i$ logical qubits into $n_i$ physical qubits. There are two procedures for the construction of the concatenated quantum codes \cite{gaitan2008quantum}, depending on whether $n_o$ is divisible by $k_i$ or not. 

The following procedure is suitable when $n_o$ is divisible by $k_i$.
\begin{enumerate}
    \item Encode $k_o$ logical qubits into $n_o$ physical qubits.
    \item Since $n_o$ is divisible by $k_i$, $n_o$ qubits will be partitioned into $\frac{n_o}{k_i}$ blocks with $n_i$ qubits in each block.
    \item This procedure results in a concatenated code given by 
\begin{equation}
[[n_\textsubscript{CC},k_\textsubscript{CC}]] = \biggr[\biggr[\frac{n_on_i}{k_i}, k_o \biggr]\biggr].
\label{eq:proc1}
\end{equation}
\end{enumerate}
In the case that $n_o$ is not divisible by $k_i$ the following procedure is employed:
\begin{enumerate}
    \item Input a quantum string of length $k_ik_o$ qubits. Encode each block of $k_o$ qubits into
$n_o$ qubits using code $\mathcal{C}_{o}$. This results in a string of length $n_o k_i$, i.e., $n_o$ blocks with $k_i$ qubits in each block.
\item Encode $k_i$ qubits in each block into $n_i$ qubits, resulting in $n_on_i$ qubits.
\item  This procedure results in a
\begin{equation}
[[n_\textsubscript{CC},k_\textsubscript{CC}]] = [[n_on_i, k_ok_i]].
\label{eq:proc2}
\end{equation}
\end{enumerate}
In the present work with a hybrid scenario, the outer code can be a QECC and the inner code a DFS or vice versa.


\section{Two schemes for hybrid two-level concatenation }\label{sec:2scheme}
In this work, we consider two-level code concatenation, where one layer employs a QECC and the other a DFS. The encoding that is first applied (last to be decoded) is the outer layer and the encoding that is applied last (decoded first) is the inner layer. We consider two orderings of concatenation: QECC being the outer code and DFS the inner one (QD code) and, conversely, DFS being the outer code and QECC the inner one (DQ code).

Throughout the paper, we will denote a $[[n, k]]$ code by $[[n, k]]_Q$ if it embeds $k$ logical qubits in $n$ physical qubits of a QECC code word and by $[[n, k]]_D$ if it embeds $k$ logical qubits in $n$ physical qubits of a DFS.

It is natural to consider correlated errors when DFS is used. We now consider a restricted correlated noise model, namely, a hybrid independent-correlated error model, where correlation across qubits is allowed within the blocks of the innermost layer of concatenation, but there are no cross-block correlations \cite{rahn2002exact}. This is more general than the independent-error model, but more restricted than the most general correlated noise.

Let the probability of the single qubit without error and with errors be given by
\begin{align}
    p_0&= 1-p,\nonumber\\
     p_1&= p.
     \label{eq:singlequbitprob}
\end{align}
Within each block, the conditional probabilities are expressed as
\begin{align}
    p_{(i|j)}=(1-\mu)p_{i}+\mu \delta_{i,j},
    \label{eq:conditionalprobgeneralformula}
\end{align}
where $i,j={0,1}$ and $\mu$ is the strength of the correlation.
Equation (\ref{eq:conditionalprobgeneralformula}) implies that
\begin{subequations}
\begin{align}
    p_{(0|0)} &= (1 - \mu) (1 - p) + \mu, \label{eq:19a}\\
    p_{(0|1)} &= (1 - \mu) (1 - p), \label{eq:19b}\\
    p_{(1|0)} &= (1 - \mu) p, \label{eq:19c}\\
    p_{(1|1)} &= (1 - \mu) p + \mu.
    \end{align}
\label{eq:conditionalprobabilitites}
    \end{subequations}
The following result gives the recursion relation that determines the failure probability $p_F^{\rm CC}$ of the concatenated code in this model.

\begin{theorem}
Let $\mathcal{C}_{0}, \mathcal{C}_{1}, \mathcal{C}_2, \ldots, \mathcal{C}_t$ denote a sequence of $(t+1)$ codes that are concatenated, with 0 labeling the outermost and $t$ the innermost. The errors are assumed to be subject to the independent-correlated error model characterized by the parameters $p$ and $\mu$.
Denoting the stand-alone failure probability of the code labeled $\mathcal{C}_j$ by $p_F^{j}(\mu,p)$, the failure probability for the concatenated code is 
\begin{equation}
  p_F^{\rm CC}(\mu,p) = p_F^{(0)}\left[0,p_F^{(1)}\left[0,\ldots[p_F^{(t-1)}[0, p_F^{(t)}(\mu,p)]]{\ldots}\right]\right].
  \label{eq:ccerror}
\end{equation}
\label{thm:ccerror}
\end{theorem}
\begin{proof}
For $1 \le j \le t-1$, given the block failure probability $p_F^{j}$ of each block at level $j$, this error is propagated to the next outer level $(j-1)$ as the block failure probability given by $p_F^{(j-1)}[\mu_{j}, p_F^{(j)}] \equiv p_F^{(j-1)}[0, p_F^{(j)}]$, the equality following from the assumption of the absence of cross-block correlation. Proceeding thus recursively, we reach the $(t-1)$th layer for which the next inner layer is the physical qubits: Thus $p_F^{(t-1)} = p_F^{(t-1)}[0, p_F^{(t)}] \equiv p_F^{t-1}[0,p_F^{(t)}(\mu,p)]$. It follows that the error in the $t$-layer concatenated code is given by $p_F^{\rm CC}(p)$ in Eq. (\ref{eq:ccerror}).
\end{proof} 

As an illustration of Theorem \ref{thm:ccerror}, consider the case of concatenating two codes $\mathcal{C}_{o}$ and $\mathcal{C}_{i}$ denoting the outer and inner codes, with failure probabilities $p_F^{(o)}(\mu,p)$ and $p_F^{(i)}(\mu,p)$. Then the failure probability of the concatenated scheme in the hybrid independent-correlated model is
\begin{equation}
  p_F^{\rm CC}(\mu,p) = p_F^{(o)}[0,p_F^{(i)}(\mu,p)].
  \label{eq:ccerroroutin}
\end{equation}
If cross-block correlations are allowed, then mathematically the simplest scenario is one where each level has an independent correlation parameter $\mu_j$ ($j \ne t$), corresponding to block-block correlation, related to possible correlations that arise during decoding. If instead we generalize by simply removing the prohibition on cross-block correlations, then the dependence of the correlation parameter at the $j$th level on the basic parameters $\mu$ and $p$ can in general be involved and not lead to a simple recursion formula in the manner of Eq. (\ref{eq:ccerror}).

To see this, consider a $[[6,1]]$ code obtained with a $[[3,1]]$ bit-flip QECC as the outer layer and the above $[[2,1]]$ DFS as the inner layer. For this QD configuration, the blocks are $(1a,1b)$, $(2a,2b)$, and $(3a,3b)$. We find the cross-block correlation to be
\begin{widetext}
\begin{align}
   p_{(\underline{0}_{1a1b}|\underline{0}_{2a2b})} &= p_{(II \vee XX | II \vee XX)} = \frac{p_{(II_{1a1b} \vee XX_{1a1b} \wedge II_{2a2b} \vee XX_{2a2b})}}{p_{(II_{2a2b} \vee XX_{2a2b})}}
   \nonumber \\
      &=\frac{1 - (1 - p) p (1 - \mu) \bigg[4 - 4 p (-1 + \mu)^2 + 
    4 p^2 (-1 + \mu)^2 + (-1 + \mu) \mu\bigg]}{[(1 - \mu)(1 - p) + \mu] (1 - p) + [(1 - \mu)p + \mu] p},
   \label{eq: correlatederrorterms}
\end{align}
\end{widetext}
which clearly does not conform to a simple recursive formula along the lines of Eq. (\ref{eq:19a}). For the remaining paper, we restrict the discussion to the hybrid model, where we assume vanishing cross-block correlations.

We quantify the performance of the concatenated code using entanglement fidelity \cite{schumacher1996sending, nielsen1996entanglement}, which in the present context can be expressed by the formula \cite{gaitan2008quantum, cafaro2010quantum}
\begin{equation}
  F_e = 1-p_F^{\rm CC}(\mu,p),
  \label{eq:EF}
\end{equation}
where $p$ is the error probability on an individual qubit. In the following two Sections, we study examples of concatenated codes and their performance under independent and independent-correlated noise.

We recollect that in a standalone code, errors $E_{a}$ and $E_b$ are degenerate if $E_{a}E_{b} \in \mathcal{S}$, where $\mathcal{S}$ is the stabilizer. Equivalently,  $\forall_{\ket{w} \in \mathcal{C}} E_{a}\ket{w} = E_{b}\ket{w}$. The DFS code has a special kind of degeneracy whereby $\forall_a E_a \in \mathcal{S}$, thereby making the errors $E_a$ passively correctable. Note that the passively correctable errors may equivalently be referred to as stabilizer errors \cite{lidar2001decoherenceII}. This element of degeneracy with passivity is inherited by the concatenation that it forms a part of. The structure of this inherited feature is discussed in the following result.

\begin{theorem}
The concatenation of a QECC and DFS code results in a degenerate code that splits into actively and passively correcting parts, with the degeneracy impacting either part, and leading to degenerate errors as well as degenerate stabilizers.
\label{thm:splittwoparts}
\end{theorem}
\begin{proof}
We denote the stabilizer groups for QECC and DFS by $\langle Q_{1}, Q_{2}, Q_{3} \ldots Q_q\rangle$ and $\langle D_{1}, D_{2}, D_{3} \ldots D_d\rangle$, respectively. 

In the QD concatenation, irrespective of whether $k_i \mid n_o$ or $k_i \nmid n_o$, a set of stabilizers $\langle D_{1}, D_{2}, D_{3} \ldots D_d\rangle$ is assigned to each ($n_i$-sized) block of the concatenated code. By design, these will constitute a set of passive, degenerate stabilizers, which will be identical to passively corrected errors. If there are $n$ blocks in the inner layer, then the number of degenerate passive stabilizer generators will be $n \times d$. 

The other part of the stabilizer can be formed by encoding the elements of $\langle Q_{1}, Q_{2}, Q_{3} \ldots Q_q\rangle$ in terms of the logical operations of the DFS code. Here the active nature of the QECC is inherited, in that the active stabilizers $Q_1, Q_2, \ldots$ are applied when the encoded quantum information has been decoded at the inner layer and brought to the outer layer. This leads to $q$ stabilizers modulo degeneracy. The degeneracy in these active stabilizers arises because of possible multiplicity in the logical operations of the DFS code.

In the DQ concatenation, irrespective of whether $k_i \mid n_o$ or $k_i \nmid n_o$, a set of stabilizers $\langle Q_{1}, Q_{2}, Q_{3} \ldots Q_q\rangle$ is assigned to each ($n_i$-sized) block of the concatenated code. By design, these will constitute a set of active stabilizers. The number of (non-degenerate) active stabilizer generators will be $n \times q$. The other part of the stabilizer can be formed by encoding the elements of $\langle D_{1}, D_{2}, D_{3} \ldots D_d \rangle$ in terms of the logical operations of the QECC. This leads to $d$ degenerate, passive stabilizers.
\end{proof}
The following three sections illustrate Theorem \ref{thm:splittwoparts} in the context of the two-level concatenation of a $[[2,1]]$ DFS code (\ref{eq:dfslogic}) and a $[[3,1]]$ bit-flip QECC or a $[[5,1]]$ QECC.


\section{The [[6,1]] QD and DQ codes \label{sec-conccodes6}}

Consider the concatenation of the $[[3,1]]_Q$ bit-flip code and the $[[2,1]]_D$ DFS code (\ref{eq:dfslogic}). In previous works, either the QD or DQ schemes were considered individually \cite{clemens2004quantum, cafaro2011quantifying, cafaro2011concatenation} but not comparatively, as we do here.

\begin{figure}[htp]
    \centering
     \begin{subfigure}[b]{0.225\textwidth}
         \centering
         \includegraphics[width=\textwidth]{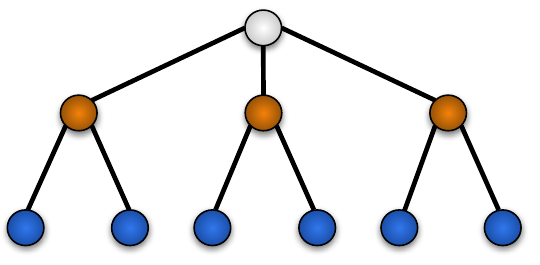}
         \caption{}
         \label{fig:61QD}
     \end{subfigure}
     \hfill
    \begin{subfigure}[b]{0.225\textwidth}
         \centering
         \includegraphics[width=\textwidth]{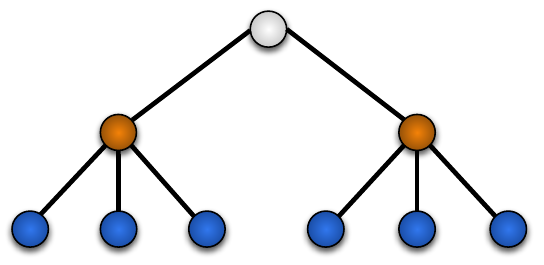}
         \caption{}
         \label{fig:61DQ}
     \end{subfigure}
    \caption{ The [[6,1]] block structure: (a) QD configuration and (b) DQ configuration. The dark orange circles represent the outer code blocks, while the blue circles represent the inner code qubits.}
    \label{fig-block-6-qdq}
    \end{figure}
First we consider the QD scheme, where the former is the outer and the latter the inner code [Fig. \ref{fig:61QD}]. By Eq. (\ref{eq:proc1}), this results in the concatenated $ [[6,1]]_{QD}$ code, for which the logical codewords are
\begin{align}
\ket{0} &\longrightarrow  \ket{0}_D \ket{0}_D \ket{0}_D
= \frac{1}{2\sqrt{2}}(\ket{00}+\ket{11})^{\otimes3}, \nonumber \\
\ket{1} &\longrightarrow  \ket{1}_D \ket{1}_D \ket{1}_D
= \frac{1}{2\sqrt{2}}(\ket{01}+\ket{10})^{\otimes3}.
\label{eq:QD[6,1]}
\end{align}
The circuits for encoding in the $[[6,1]]$ QD or DQ states (Fig. \ref{fig-block-6-qdq}) are given in Figs. \ref{fig:circuit61QD} and \ref{fig:circuit61DQ}, respectively.
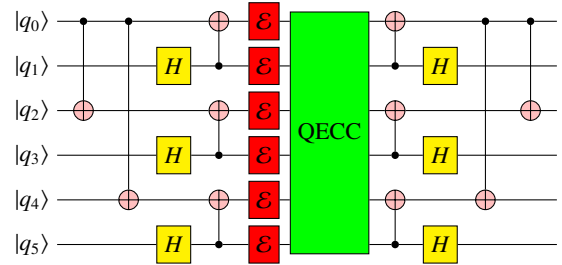
\begin{figure} 
\centering
\begin{tikzpicture}
\begin{yquant}
qubit {$\ket{\reg_{\idx}}$} q[6];
[fill=pink]
cnot q[2] | q[0];
[fill=pink]
cnot q[4] | q[0];
[fill=yellow]
h q[1,3,5];
[fill=pink]
cnot q[0] | q[1];
[fill=pink]
cnot q[2] | q[3];
[fill=pink]
cnot q[4] | q[5];
[fill=red]
box {$\mathcal{E}$} q;
[fill=green]
box {QECC} (q);
[fill=pink]
cnot q[4] | q[5];
[fill=pink]
cnot q[2] | q[3];
[fill=pink]
cnot q[0] | q[1];
[fill=yellow]
h q[5,3,1];
[fill=pink]
cnot q[4] | q[0];
[fill=pink]
cnot q[2] | q[0];
\end{yquant}
\end{tikzpicture}
\caption{ Circuit to implement the $[[6,1]]$ QD code. Here qubit $q_0$ stores the encoded information while the remaining (ancillary) qubits $q_1$-$q_5$ start in the state $\ket0$. The circuit shows the sequential application of the encoding operation, possible bit-flip errors on the physical qubits, the error correction step, and finally a decoding operation, which reverses the encoding to restore the quantum information at qubit $q_0$.
}
\label{fig:circuit61QD}
\end{figure}

In the DQ scheme, we concatenate in the reverse order: the $[[3,1]]_Q$ bit flip code as the inner code with the $[[2,1]]_D$ DFS code as the outer code [Fig. \ref{fig:61DQ}]. By Eq. (\ref{eq:proc1}), this results in the concatenated $[[6,1]]_{DQ}$ code, with the logical code words
\begin{align}
\ket{0} &\longrightarrow \ket{00}_Q+\ket{11}_Q
= \frac{1}{\sqrt{2}}\ket{000}^{\otimes2}+\ket{111}^{\otimes2}, \nonumber\\
\ket{1}&\longrightarrow \ket{01}_Q+\ket{10}_Q =\frac{1}{\sqrt{2}}\ket{000111}+\ket{111000}.
\label{eq:onelogicsixonedq}
\end{align}
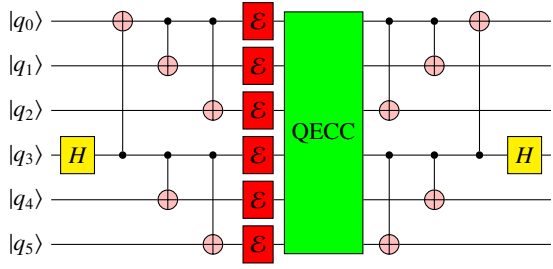
\begin{figure} 
\centering
\begin{tikzpicture}
\begin{yquant}
qubit {$\ket{\reg_{\idx}}$} q[6];
[fill=yellow]
h q[3];
[fill=pink]
cnot q[0] | q[3];
[fill=pink]
cnot q[1] | q[0];
[fill=pink]
cnot q[2] | q[0];
[fill=pink]
cnot q[4] | q[3];
[fill=pink]
cnot q[5] | q[3];
[fill=red]
box {$\mathcal{E}$} q;
[fill=green]
box {QECC} (q);
[fill=pink]
cnot q[5] | q[3];
[fill=pink]
cnot q[4] | q[3];
[fill=pink]
cnot q[2] | q[0];
[fill=pink]
cnot q[1] | q[0];
[fill=pink]
cnot q[0] | q[3];
[fill=pink]
[fill=yellow]
h q[3];
\end{yquant}
\end{tikzpicture}
\caption{ Circuit to implement the $[[6,1]]$ DQ code. Here qubit $q_0$ stores the encoded information while the remaining (ancillary) qubits $q_1$-$q_5$ start in the state $\ket0$. The circuit shows the sequential application of the encoding operation, possible bit-flip errors on the physical qubits, the error correction step, and finally a decoding operation, which reverses the encoding to restore the quantum information at qubit $q_0$.}
\label{fig:circuit61DQ}
\end{figure}

\subsection{Equivalence classes of correctable Pauli errors}
For the QD scheme, we note that the errors protected at the inner layer are
\begin{align}
\{II,XX\}^{\otimes3}
\label{eq:61QDinner}
\end{align}
and the errors protected at the outer layer are
\begin{align}
\{I_D I_D I_D, X_D I_D I_D, I_D X_D I_D, I_D I_D X_D\},
\label{eq:61QDouter}
\end{align}
where the subscript $D$ refers to error operations of the DFS.
For the $[[2,1]]_D$ code, logical bit flip is defined as 
$X_D: \ket{0}_D \longleftrightarrow \ket{1}_D$. Letting $I_{D}$ represent any of correctable errors, and $X_{D}$, $Y_{D}$, and $Z_{D}$ represent the logical operators for the $[[2,1]]_{D}$ code, we have 
\begin{subequations}
\begin{eqnarray}
    I_D & \rightarrow II,XX , \label{eq:61QDmulti_a}\\
     X_D & \rightarrow XI, IX, \label{eq:61QDmulti_b} \\
     Y_D & \rightarrow YZ, ZY, \label{eq:61QDmulti_c} \\
     Z_D & \rightarrow ZZ, -YY.
    \label{eq:61QDmulti_d}
\end{eqnarray}
\label{eq:61QDmulti}
\end{subequations}
Here we note that all four operations have a multiplicity of 2.

Each block's error in the outer code, namely, $I_D$ or $X_D$, in Eq. (\ref{eq:61QDouter}), is represented by a multiplicity (actually, 2) of realizations in the inner code, given by Eq. (\ref{eq:61QDmulti}). Thus, each of the four correctable errors in Eq. (\ref{eq:61QDouter}) has $2^3 = 8$ degenerate realizations.
Thus there are 32 errors that the $[[6,1]]_{QD}$ code can correct, and these can be arranged in an equivalence class of four sets such that all eight elements within a given set are mutually degenerate,
given by
\begin{widetext}
\begin{align}
    E_{\rm {equiv}}^{[[6,1]]_{QD}} &\equiv \Bigl\{\{IIIIII, IIIIXX, IIXXII, IIXXXX, XXXXXX, XXIIII, XXIIXX, XXXXII\}, \{XIIIII, XIIIXX, XIXXII,\nonumber\\
&  XIXXXX, IXIIII, IXIIXX, IXXXII, IXXXXX \}, \{IIXIII, IIXIXX, IIIXII, IIIXXX, XXXIII, XXXIXX,\nonumber\\
& XXIXII, XXIXXX \},\{IIIIXI, IIIIIX, IIXXXI, IIXXIX, XXIIXI, XXIIIX, XXXXXI, XXXXIX \}\Bigr\}.
\label{eq:eec4}
\end{align}
The four mutually degenerate sets in this \textit{error degeneracy equivalence class} correspond to four sets having eight degenerate elements each. The QD and DQ codes inherently possess an element of passive error correction inherited from the DFS. Here, this is reflected in the fact that elements in the first set of the equivalence class are all passively correctable. Their provenance can be attributed to the fact that all these errors correspond to the identity operation ($I_DI_DI_D$) on the outer QECC layer. Structurally, we expect that these error operators are identical to eight stabilizers of the concatenated code and thus commute with the remaining stabilizers. We will find that these eight stabilizers (or their three generators) play a \textit{passive} role in that they do not require measurement, but formally arise as a result of the concatenation.

By contrast, with regard to the DQ scheme [Fig. \ref{fig:61DQ}], 
the errors protected by the outer DFS layer are
\begin{align}
\{I_Q I_Q, X_Q X_Q\},
\label{eq:61DQouter}
\end{align}
where the subscript $Q$ refers to the error operations of the QECC.
There are 16 correctable errors that yield the first term in Eq. (\ref{eq:61DQouter}) when the inner layer is decoded, namely,
\begin{align}
\{XII,IXI,IIX,III\}^{\otimes2}.
\label{eq:61DQinner}
\end{align}
Each of these has a counterpart corresponding to the second term in Eq. (\ref{eq:61DQouter}), obtained by applying the logical NOT operation on both blocks, i.e., $XXXXXX$. Thus, for example, $XIIIXI \xrightarrow{XXXXXX} IXXXIX$.

Accordingly, we have 16 sets of two errors each, yielding 32 correctable errors in all, which can be arranged in the following equivalence class:
\begin{align}
    E_{\rm equiv}^{[[6,1]]_{DQ}} &\equiv \Bigl\{\{XIIXII, IXXIXX\}, \{XIIIXI, IXXXIX \}, \{XIIIIX, IXXXXI\}, \{XIIIII, IXXXXX \},\nonumber\\
& \{IIXXII, XXIIXX\}, \{IIXIXI, XXIXIX \}, \{IIXIIX, XXIXXI\},\{IIXIII, XXIXXX \},\nonumber\\
&\{IXIXII, XIXIXX\}, \{IXIIXI, XIXXIX \}, \{IXIIIX, XIXXXI\}, \{IXIIII, XIXXXX \},\nonumber\\
& \{IIIXII, XXXIXX\}, \{IIIIXI, XXXXIX \}, \{IIIIIX, XXXXXI\}, \{IIIIII, XXXXXX\}\Bigr\}.
\label{eq:eec16}
\end{align}
\end{widetext}
The 16 mutually degenerate sets in the equivalence class correspond to 16 sets having two degenerate elements each. As in the QD case, passively correctable errors arise here too. In this example, they correspond to the last set in Eq. (\ref{eq:eec16}). Their provenance can be attributed to the fact that they lead, \textit{without} syndrome generation, to the DFS correctable errors $I_QI_Q$ and $X_QX_Q$ in the outer layer. Here again, the passively correctable errors can be shown to coincide with passive stabilizers.

The Hamming bound \cite{ekert1996quantum} on a nondegenerate QECC requires that $n-k \ge \log_2({\rm no.~of~correctable ~errors})$. Thus we can define the metric of \textit{Hamming efficiency} as a rough guide on how efficiently the correction works:
\begin{equation}
    \varphi = \frac{\log_2\left({\rm no.~of~correctable ~errors}\right)}{n-k}.
    \label{eq:density}
\end{equation}
From Eqs. (\ref{eq:eec4}) and (\ref{eq:eec16}), we have $\Vert E_{\rm equiv}^{[[6,1]]_{QD}} \Vert =4 \times 8 = 32$ and $\Vert E_{\rm equiv}^{[[6,1]]_{DQ}} \Vert =16 \times 2 = 32$ as the number of elements in the equivalence class for respective QD and DQ codes. Here we adopt the notational convention wherein the symbol $\vert E_{\rm equiv}^{CC} \vert$ denotes the cardinality of the equivalence class $E_{\rm equiv}^{CC}$ (i.e., the number of sets in the class), whereas $\Vert E_{\rm equiv}^{CC} \Vert$ denotes the total number of elements in the equivalence class. Thus, for both the above $[[6,1]]$ concatenated codes, $\varphi = \frac{\log_2(32)}{5} = 1$, which is appropriate for a perfect code. However, as the present codes are degenerate, we suggest that in this case it seems more appropriate to generalize Eq. (\ref{eq:density}) to the modified Hamming efficiency
\begin{equation}
    \varphi^{\prime}  = \frac{\log_2\left(|{\rm ~error~equivalence~class}|\right)}{n-k}.
    \label{eq:density+}
\end{equation}
This quantifies the number of distinguishable correctable errors as a fraction of non-coding bits. From Eqs. (\ref{eq:eec4}) and (\ref{eq:eec16}), $\vert E_{\rm equiv}^{[[6,1]]_{QD}} \vert = 4$ and $\vert E_{equiv}^{[[6,1]]_{DQ}} \vert = 16$, respectively. Thus $\varphi^{\prime} = \frac{2}{5}$ and $\varphi^{\prime} = \frac{4}{5}$ for the $[[6,1]]_{QD}$ and $[[6,1]]_{DQ}$ codes, respectively. 

The above examples illustrate the following result, which is generally true for QD and DQ concatenations.
\begin{theorem}
All passively correctable errors of a given QD or DQ code are mutually degenerate, i.e., they constitute a single set in the error degeneracy equivalence class.
\label{thm:setedec}
\end{theorem}
\begin{proof}
In the QD case, the passively correctable errors of the concatenated code decode to the identity error in the outer QECC layer. In the DQ case, these passively correctable errors decode syndromelessly to the correctable errors in the outer DFS layer. In the QD and DQ cases, therefore the passive correctability criterion fulfills the requirement to constitute an element of the error degeneracy equivalence class.
\end{proof}
As an illustration of Theorem \ref{thm:setedec}, in a given equivalence class of correctable errors for a concatenation involving DFS as one layer of protection, all elements in a set or none will possess a DFS-like structure.

\begin{figure}[htp]
    \centering
    \includegraphics[width=8cm]{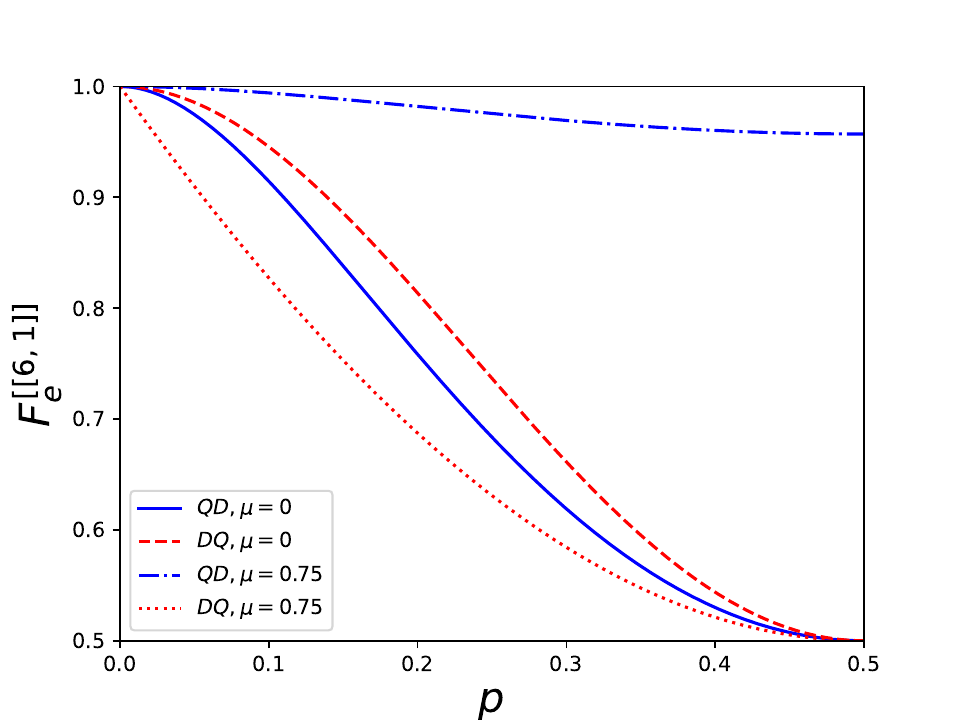}
    \caption{(Color online) Entanglement fidelities of the codes $[[6,1]]_{QD}$ and $[[6,1]]_{DQ}$ for the independent-error model ($\mu=0$) and the restricted correlated error model ($\mu=0.75$, without cross-block correlation), respectively, depicted as the blue solid, red dashed, blue dash-dotted, and red dotted plots. The plots illustrate the idea that QD codes outperform DQ ones for sufficiently strong correlation of the noise. }
     \label{fig:ef61}
\end{figure}

\subsection{Performance of the [[6,1]] QD and DQ codes under correlated noise}
We now discuss the performance of the $[[6,1]]$ QD and DQ codes in terms of the failure probability for the hybrid independent-correlated noise model.
In a correlated error model, the stand-alone failure probability for $[[3,1]]_{Q}$ and $[[2,1]]_{D}$ codes are given by
\begin{subequations}
\begin{align}
    p_F^{[[3,1]]_{Q}}(\mu,p) &=p_{(0|0)}p_{(0|0)}p_0 + p_{(1|0)}p_{(0|0)}p_0 \nonumber\\
    &+ p_{(0|1)}p_{(1|0)}p_0 + p_{(0|0)}p_{(0|1)}p_1  \nonumber \\
    &= (3p^2-2p^3)(\mu -1)^2-p(\mu - 2)p,
    \label{eq:pf3Qcorrelated:31}\\
    p_F^{[[2,1]]_{D}}(\mu,p) &=p_{(1|0)}p_0 + p_{(0|1)}p_1  \nonumber \\
    &= 2 (1 - p) p (1 - \mu).
    \label{eq:pf3Qcorrelated:21}
\end{align}
 \label{eq:pf3Qcorrelated}
\end{subequations}
By Eqs. (\ref{eq:ccerroroutin}) and (\ref{eq:pf3Qcorrelated}), the failure probability of the $[[6,1]]_{QD}$ code is
\begin{align}
    p_F^{[[6,1]]_{QD}} &= p_F^{[[3,1]]_{Q}}\left[0,p_F^{[[2,1]]_{D}}(\mu,p)\right]\nonumber\\
   &= 3\left[p_F^{[[2,1]]_{D}}(\mu,p)\right]^2\left[1-p_F^{[[2,1]]_{D}}(\mu,p)\right]\nonumber\\
   &+\left[p_F^{[[2,1]]_{D}}(\mu,p)\right]^3,
    \label{eq:61pfQDcorrelated}
    \end{align}
whereas the failure probability of the $[[6,1]]_{DQ}$ code is
\begin{align}
    p_F^{[[6,1]]_{DQ}} &= p_F^{[[2,1]]_{D}}\left[0,p_F^{[[3,1]]_{Q}}(\mu,p)\right]\nonumber\\
    &=2p_F^{[[3,1]]_{Q}}(\mu,p)\left[1-p_F^{[[3,1]]_{Q}}(\mu,p)\right].
    \label{eq:61pfDQcorrelated}
\end{align}

The above two failure probabilities lead to roughly similar behavior, with the failure probability for the concatenated code attaining the maximum for $p=\frac{1}{2}$, as should be the case. The entanglement fidelities [according to Eq. (\ref{eq:EF})] for the two concatenated codes, namely, $F_e^{[[6,1]]_{QD}}$ and $F_e^{[[6,1]]_{DQ}}$, are plotted in Fig. \ref{fig:ef61} both for the regime of independent errors ($\mu=0$) and for errors with high correlation ($\mu=0.75$). In the former case, the DQ code outperforms the QD code, because the inner layer of the QD code fails to correct many errors that are independent.  
In the latter case, the QD code outperforms the DQ code, because the inner layer of the DQ code fails to correct many errors that are correlated.


\section{The [[10,1]] QD and DQ codes \label{sec-conccodes10}}

Here the $[[5,1]]$ Knill-Laflamme QECC \cite{knill2001benchmarking} is concatenated with the $[[2,1]]_{D}$ code. Such a concatenation may be useful in situations where independent and correlated errors coexist \cite{boulant2005experimental}. When the latter errors are stronger than the independent errors, the QD scheme is preferable, whereas the DQ scheme is preferable when independent errors are stronger than correlated errors.

The $[[10,1]]_{QD}$ code is depicted in Fig. \ref{fig:101QD}.
The logical code words for this code are
\begin{align}
\ket{0}&\longrightarrow  \frac{1}{4}(\ket{00000}_D+\ket{10010}_D+\ket{01001}_D+\ket{10100}_D\nonumber\\
&+\ket{01010}_D-\ket{11011}_D-\ket{00110}_D-\ket{11000}_D\nonumber\\
&-\ket{11101}_D-\ket{00011}_D-\ket{11110}_D-\ket{01111}_D\nonumber\\
&-\ket{10001}_D-\ket{01100}_D-\ket{10111}_D+\ket{00101}_D ),\nonumber\\
\ket{1}&\longrightarrow \frac{1}{4}(\ket{11111}_D +\ket{01101}_D+\ket{10110}_D+\ket{01011}_D\nonumber\\
&+\ket{10101}_D-\ket{00100}_D-\ket{11001}_D-\ket{00111}_D\nonumber\\
&-\ket{00010}_D-\ket{11100}_D-\ket{00001}_D-\ket{10000}_D\nonumber\\
&-\ket{01110}_D-\ket{10011}_D-\ket{01000}_D-\ket{11010}_D ),
\label{eq:QD[10,1]}
\end{align}
where the subscript $D$ indicates that there is an inner layer of encoding where each qubit is replaced by the corresponding encoded qubit of the DFS. Thus, for example, 
\begin{align}
    \ket{00000}_D&=\ket{0}_D \ket{0}_D \ket{0}_D \ket{0}_D \ket{0}_D\nonumber\\
    &=(\ket{00}+\ket{11})^{\otimes5},
    \nonumber
\end{align}
and so forth.

\begin{figure}[htp]
    \centering
     \begin{subfigure}[b]{0.225\textwidth}
         \centering
         \includegraphics[width=\textwidth]{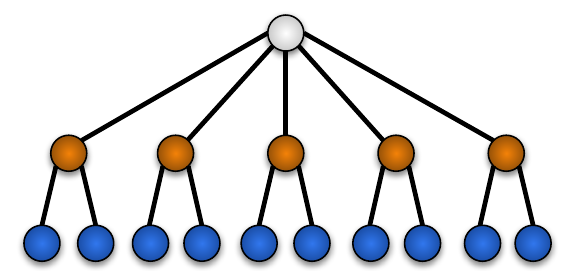}
         \caption{}
         \label{fig:101QD}
     \end{subfigure}
     \hfill
    \begin{subfigure}[b]{0.225\textwidth}
         \centering
         \includegraphics[width=\textwidth]{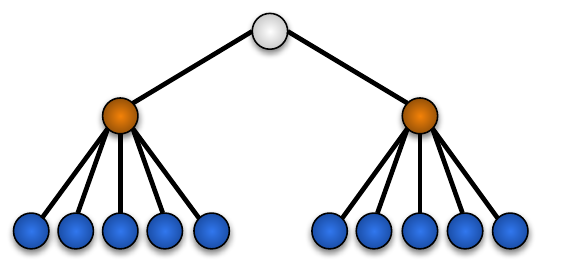}
         \caption{}
         \label{fig:101DQ}
     \end{subfigure}
    \caption{The [[10,1]] block structure: (a) QD configuration and (b) DQ configuration. The dark orange circles represent the outer code blocks, while the blue circles represent the inner code qubits.}
    \label{fig-block-10-qdq}
    \end{figure}

The DQ scheme of the $[[10,1]]$ code is depicted in Fig. \ref{fig:101DQ}, for which the logical code words for this ten-qubit DQ code are
\begin{align}
\ket{0}&\longrightarrow |00\rangle_Q + |11\rangle_{Q},\nonumber\\
\ket{1}&\longrightarrow \ket{01}_Q + \ket{10}_{Q},
\label{{eq:DQ[10,1]}}
\end{align}
where $\ket{0}_Q$ and $\ket{1}_Q$ are the logical code words for the five-qubit QECC.

\begin{figure}[htp]
    \centering
    \includegraphics[width=8cm]{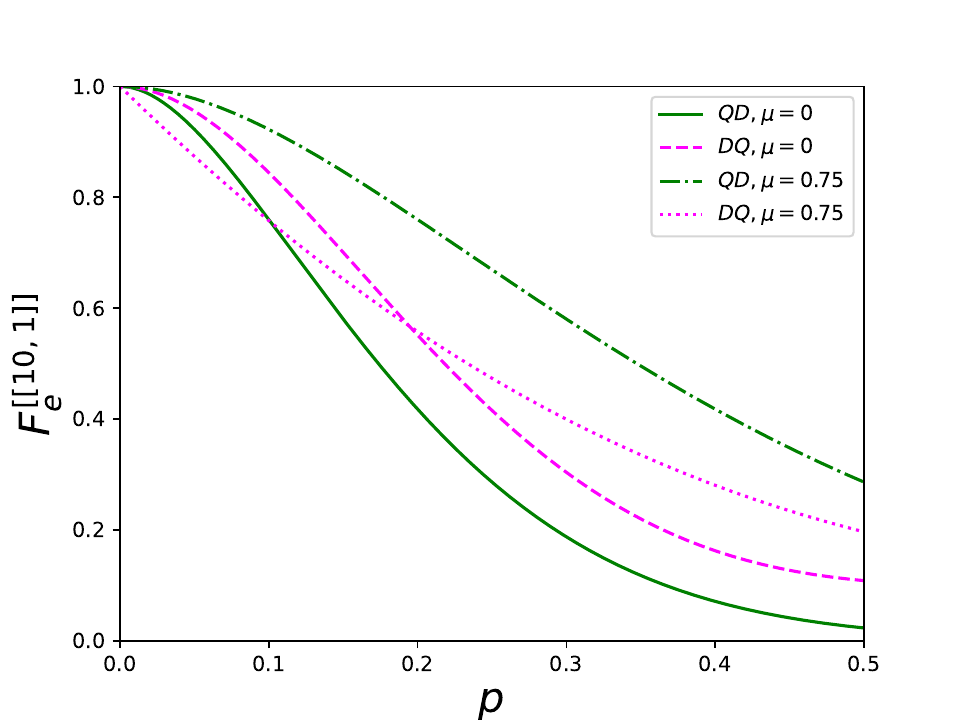}
    \caption{(Color online)  Entanglement fidelities of the concatenated codes $[[10,1]]_{QD}$ and $[[10,1]]_{DQ}$ under the independent-error model ($\mu=0$) and the restricted-correlated error model ($\mu=0.75$), respectively, depicted as the green solid, magenta dashed, green dash-dotted and magenta dotted plots.  As in the six-qubit case, for sufficiently correlated error, QD outperforms DQ and the reverse when the independent-error model is used.}
     \label{fig:ef101}
\end{figure}

\subsection{Equivalence classes of correctable Pauli errors}
The correctable errors here can be arranged in an equivalence class, given by the collection of 16 sets
\begin{align}
    E_{\rm equiv}^{[[10,1]]_{QD}} \equiv &~~ \{[IIIII]_D,[XIIII]_D,[IXIII]_D,[IIXII]_D,\nonumber \\ &~~[IIIXI]_D,[IIIIX]_D,[YIIII]_D,[IYIII]_D, \nonumber \\
    &~~ [IIYII]_D,[IIIYI]_D,[IIIIY]_D, [ZIIII]_D,\nonumber\\
    &~~  [IZIII]_D, [IIZII]_D,[IIIZI]_D, [IIIIZ]_D
    \},
    \label{eq:101QDequivA}
\end{align}
which correspond the 16 errors correctable by the outer code. Here $[\cdots]_D$ denotes the expansion of the five-qubit operator string by insertion of corresponding logical operation in the inner code, e.g., $[IIIII]_D \equiv \{II, XX\}^{\otimes 5}$, yielding 32 operators with a length ten-qubit operators. Thus Eq. (\ref{eq:101QDequivA}) represents an equivalence class of $\vert E_{\rm equiv}^{[[10,1]]_{QD}}\vert = 16$ sets having 32 degenerate elements each, which yields $\Vert E_{\rm equiv}^{[[10,1]]_{QD}}\Vert = 16 \times 32 = 512$ elements in all. Here all elements in the set $[IIIII]_D$ are passively correctable, in the manner of those in the set $[III]_D$ of $E_{\rm equiv}^{[[6,1]]_{QD}}$ [Eq. (\ref{eq:eec4})]. In Sec. \ref{sec:stab10} we will show that these errors are generated by five passive stabilizer generators. Therefore, for the $[[10,1]]_{QD}$ concatenated code, the Hamming efficiency $\varphi = \frac{\log_2(512)}{9} =1$ and the modified Hamming efficiency $\varphi^{\prime} = \frac{\log_2(16)}{9} = \frac{4}{9}$.

For the $[[10,1]]_{DQ}$ code, in place of Eq. (\ref{eq:101QDequivA}), we have $E_{\rm equiv}^{[[10,1]]_{DQ}} \equiv \{[II]_Q,[XX]_Q$\}. There are 16 Pauli errors that are correctable and thus yield $I_Q$ after the inner layer has been decoded. Correspondingly, there are 16 Pauli errors that result in $X_Q$ after the inner layer has been decoded. These are just the correctable errors to which the logical NOT operator $XXXXX$ has been applied, e.g., $IIYII \rightarrow XXZXX$. Thus there are $16^2$ sets consisting of a pair of Pauli errors in the set $E_{\rm equiv}^{[[10,1]]_{DQ}}$. We thus have $\vert E_{\rm equiv}^{[[10,1]]_{DQ}} \vert =16^2 = 256$ as the number of sets in the class, and $\Vert E_{\rm equiv}^{[[10,1]]_{DQ}} \Vert =16^2 \times 2 = 512$ as the total number of elements in the equivalence class. Accordingly, the Hamming efficiency is $\varphi = \frac{\log_2(512)}{9} = 1$ and the modified Hamming efficiency is $\varphi^{\prime} = \frac{\log_2(256)}{9} = \frac{8}{9}$. In the manner of the $[[6,1]]_{DQ}$ case [Eq. (\ref{eq:eec16})], here too there are only two passively correctable errors, ($I^{\otimes 10}$ and $X^{\otimes 10}$), generated by a passive stabilizer generator (Sec. \ref{sec:stab10}). 

\subsection{Performance of the [[10,1]] QD and DQ codes under correlated noise}
In a correlated error model, the stand-alone failure probability for the $[[5,1]]_{Q}$ and $[[2,1]]_{D}$ code is calculated as
\begin{widetext}
\begin{subequations}
\begin{align}
    p_F^{[[5,1]]_{Q}}(\mu,p) & = 1-p_C^{[[5,1]]_{Q}}(\mu,p)\nonumber\\
    &=1 - 3 (1 - p)^2 p (1 - \mu)^2 \left[(1 - p) (1 - \mu) + \mu \right]^2- 2 (1 - p) p (1 - \mu) \left[(1 - p) (1 - \mu) + \mu \right]^3\nonumber\\
    & - (1 - p) \left[(1 - p) (1 - \mu) + \mu \right]^4,  \label{eq:pf5Qcorrelated:51}\\
    p_F^{[[2,1]]_{D}}(\mu,p) &= 1-p_C^{[[2,1]]_{D}}(\mu,p)\nonumber\\
    &= 1 - (1 - p) \left[(1 - p) (1 - \mu) + \mu \right ]-\frac{p}{3}(\frac{p}{3} (1 - \mu) + \mu),
\label{eq:pf5Qcorrelated:21}
\end{align}
\label{eq:pf5Qcorrelated}
\end{subequations}
\end{widetext}
where $p_C^{[[5,1]]_{Q}}$ and $p_C^{[[2,1]]_{D}}$ are the probabilities of the correctable errors for the $[[5,1]]_{Q}$ and $[[2,1]]_{D}$ codes, respectively. Here $p_0=1-p$ and $p_1=p_2=p_3=\frac{p}{3}$. The DFS failure probability (\ref{eq:pf5Qcorrelated:21}), which includes $Y$ and $Z$ errors, may be contrasted with Eq. (\ref{eq:pf3Qcorrelated:21}), which includes only $X$ errors.
By Eqs. (\ref{eq:ccerroroutin}) and (\ref{eq:pf5Qcorrelated}), the failure probability of the $[[10,1]]_{QD}$ code is
\begin{align}
    p_F^{[[10,1]]_{QD}}(\mu,p) &= p_F^{[[5,1]]_{Q}}\left[0,p_F^{[[2,1]]_{D}}(\mu,p)\right]\nonumber\\
   &=1 - \left[1 - p_F^{[[2,1]]_{D}}(\mu,p)\right]^5 \nonumber\\
   &- 5 \left[1 - p_F^{[[2,1]]_{D}}(\mu,p)\right]^4 p_F^{[[2,1]]_{D}}(\mu,p),
    \label{eq:101pfQDcorrelated}
\end{align}
whereas the failure probability of the $[[10,1]]_{DQ}$ code is
\begin{align}
    p_F^{[[10,1]]_{DQ}}(\mu,p) &= p_F^{[[2,1]]_{D}}\left[0,p_F^{[[5,1]]_{Q}}(\mu,p)\right]\nonumber\\
    &=\frac{2}{3}\left[p_F^{[[5,1]]_{Q}}(\mu,p)\right]\left(1-p_F^{[[5,1]]_{Q}}(\mu,p)\right).
    \label{eq:101pfDQcorrelated}
    \end{align}
Employing Eq. (\ref{eq:EF}), we obtain the entanglement fidelities $F_e^{[[10,1]]_{QD}}$ and $F_e^{[[10,1]]_{DQ}}$ for the ten-qubit QD and DQ, respectively. These fidelities as a function of single-qubit error probability $p$ are plotted in Fig. \ref{fig:ef101}. We find that for sufficiently low correlation of the noise, the DQ code outperforms the QD code and vice versa when the correlation is sufficiently large, as discussed earlier in the context of the $[[6,1]]$ code.

\section{Generator structure of the concatenated code}\label{sec-stabilizergens}
The stabilizer of the concatenated code is a concatenation of the stabilizers of the constituent codes. A well known example is the Shor code, whose stabilizer generators can be constructed by concatenating the three-qubit phase-flip and bit-flip codes \cite{gaitan2008quantum}. A DFS code is a kind of stabilizer code \cite{lidar1999concatenating, lidar2001decoherenceII} in that the code space is stabilized by commuting operators and the errors fulfill the error-correcting condition (\ref{errorcorrectingcriterion2}). However, unlike with a conventional stabilizer code  \cite{gottesman1997stabilizer}, the stabilizer generators commute with all error operators, being identical to the error operators. This is reflective of the fact that DFS is a passive rather than active error-correction scheme. Therefore, in a broader sense, we expect that the QD and DQ concatenations also possess a stabilizer structure. 

On the other hand, the DFS introduces novel elements into the stabilizer structure, resulting in the degeneracy of the correctable errors and also, interestingly, degeneracy of the stabilizers. Importantly, the stabilizer structure divides into an active and a passive part, with corresponding errors and stabilizer operators. In the following, we illustrate the stabilizer structure of QD and DQ codes by means of the $[[6,1]]$ and $[[10,1]]$ codes considered above. 

\subsection{The [[6,1]] codes}
Here the $[[3,1]]_Q$ and $[[2,1]]_D$ codes are concatenated, with stabilizer groups given by
\begin{align}
    \mathcal{S}^{[[3,1]]_{Q}}= \langle Z^1Z^2, Z^1Z^3 \rangle
    \label{eq:generators31}
\end{align}
and
\begin{align}
    \mathcal{S}^{[[2,1]]_{D}}= \langle X^1X^2 \rangle,
    \label{eq:generators21}
\end{align}
respectively. In both the QD and DQ cases, there will be $n-k=5$ stabilizer generators.

In the QD case, we label the qubits blockwise as $(1a,1b)$,$(2a,2b)$, and $(3a,3b)$. To each block we assign a copy of stabilizer generator of the DFS (\ref{eq:generators21}), yielding three of the generators
\begin{align}
    &S_{1}=X^{1a}X^{1b},\nonumber\\
    &S_{2}=X^{2a}X^{2b},\nonumber\\
    &S_{3}=X^{3a}X^{3b}.
    \label{eq:generators61QD-1}
\end{align}
Note that the elements of the first set in Eq. (\ref{eq:eec4}) for $E_{\rm equiv}^{[[6,1]]_{QD}}$ are generated by the operators in Eq. (\ref{eq:generators61QD-1}), in keeping with the fact that these errors are passively correctable. 

It is worth mentioning that the stabilizer generators $S_1$, $S_2$, and $S_3$ in Eq. (\ref{eq:generators61QD-1}) automatically stabilize the code space and do not require measurement. To elaborate the latter point, we note that, on the one hand, the passively correctable errors are generated by these generators. For example, the first set in Eq. (\ref{eq:eec4}) is generated by $S_1$, $S_2$, and $S_3$ in Eq. (\ref{eq:generators61QD-1}). On the other hand, as for the actively correctable errors, namely, the last three sets in Eq. (\ref{eq:eec4}), they commute with each of the stabilizer operators $S_1$, $S_2$, and $S_3$. For this reason, these operators may be referred to as passive. This characteristic feature of the passive stabilizer operators, namely, of not requiring measurement, will also be found to hold in the other examples discussed below.

Interestingly, these stabilizer operators are themselves degenerate, which is a consequence of the fact that all of them are passive. Their passivity by design arises from the fact that they are stabilizers of correctable errors of the inner code, and the remaining errors correspond to logical operations at the inner layer that decode to errors corrected at the outer layer. 
By virtue of the fact that the passively correctable errors are stabilizer operators, it follows that they commute with other stabilizer generators $S_4$ and $S_5$, given below.

The remaining two generators can be obtained by encoding the qubits that make up the outer code's stabilizer in terms of logical operations of the inner code. Accordingly, from Eqs. (\ref{eq:generators31}) and (\ref{eq:61QDmulti}) we find
\begin{align}
    &S_{4}=[Z^{1}][Z^{2}],\nonumber\\
    &S_{5}=[Z^{1}][Z^{3}],
    \label{eq:generators61QD-2}
\end{align}
where the terms in the square brackets can each be encoded in two ways, according to Eq. (\ref{eq:61QDmulti_d}). 
Thus, each of the stabilizer elements in Eq. (\ref{eq:generators61QD-2}) represents an equivalence class of four generators that are degenerate in the sense that they produce the same syndrome for each correctable error.
This degeneracy of the stabilizers should be distinguished from that of the stabilizers in Eq. (\ref{eq:generators61QD-1}), which has its origin in the passivity of their error correction. 
Thus, the set of five stabilizer generators for the $[[6,1]]_{QD}$ code is given by Eq. (\ref{eq:generators61QD-1}) and selecting any one element from each of the equivalence classes $S_4$ and $S_5$ of Eq. (\ref{eq:generators61QD-2}).

In the DQ case, we label the qubits blockwise as $(1a,1b,1c)$, and $(2a,2b,2c)$. To each block we assign a copy of stabilizer generators of the QECC (\ref{eq:generators31}), yielding four of the five generators
\begin{align}
    &S_{1}=Z^{1a}Z^{1b},\nonumber\\
    &S_{2}=Z^{1a}Z^{1c},\nonumber\\
    &S_{3}=Z^{2a}Z^{2b},\nonumber\\
    &S_{4}=Z^{2a}Z^{2c}.
    \label{eq:generators61DQ-1}
\end{align}
The remaining single generator can be obtained by encoding the qubits that make up the outer code's stabilizer in terms of logical operations of the inner code. Accordingly, from Eq. (\ref{eq:generators31}) and the fact that $X_L = XXX$ for the $[[3,1]]_Q$ QECC, we find
\begin{align}
    &S_{5}=(X^{1a}X^{1b}X^{1c})(X^{2a}X^{2b}X^{2c}).
    \label{eq:generators61DQ-2}
\end{align}
Note that this generator is passive and generates the two passively correctable errors $I^{\otimes6}$ and $X^{\otimes 6}$. Summarizing, the set of five stabilizer generators for the $[[6,1]]_{QD}$ code are given by Eqs. (\ref{eq:generators61DQ-1}) and (\ref{eq:generators61DQ-2}). 

\subsection{The [[10,1]] codes}\label{sec:stab10}

The $[[5,1]]_{Q}$ and $[[2,1]]_{D}$ codes are concatenated here, and the stabilizer groups are provided by
\begin{align}
    \mathcal{S}^{[[5,1]]_{Q}}= \langle X^1Z^2Z^3X^4, X^2Z^3Z^4X^5, X^1X^3Z^4Z^5, Z^1X^2X^4Z^5 \rangle
    \label{eq:generators51}
\end{align}
and Eq. (\ref{eq:generators21}), respectively. Here there will be $n-k=9$ stabilizer generators for both QD and DQ codes.
In the QD case, the qubits are designated via blockwise labeling as $(1a,1b)$, $(2a,2b)$, $(3a,3b)$, $(4a,4b)$, and $(5a,5b)$. The assignment of a copy of the stabilizer generators of the DFS (\ref{eq:generators21}) to each block will give five stabilizer generators
\begin{align}
    &S_{1}=X^{1a}X^{1b},\nonumber\\
    &S_{2}=X^{2a}X^{2b},\nonumber\\
    &S_{3}=X^{3a}X^{3b},\nonumber\\
    &S_{4}=X^{4a}X^{4b},\nonumber\\
    &S_{5}=X^{5a}X^{5b}.
    \label{eq:generators101QD-1}
\end{align}
Note that the 32 elements of the error set $[IIIII]_{D}$ in Eq. (\ref{eq:101QDequivA}) are generated by the operators $S_j$ in Eq. (\ref{eq:generators101QD-1}), consistent with the fact that those errors are passively correctable. As in the case of the $[[6,1]]_{QD}$ code, these operators are themselves passive and thus degenerate.

The remaining four stabilizer generators are the encoded version of stabilizer generators of the five-qubit code Eq. (\ref{eq:generators51}) in block labels. These four generators, though actively error-correcting, also involve stabilizer degeneracy, this arising because of the multiplicity (\ref{eq:61QDmulti}): 
\begin{align}
    &S_{6}=[X^1][Z^2][Z^3][X^{4}],\nonumber\\
    &S_{7}=[X^{2}][Z^3][Z^4][X^{5}],\nonumber\\
    &S_{8}=[X^1][X^3][Z^4][Z^5],\nonumber\\
    &S_{9}=[Z^1][X^2][X^4][Z^5].
      \label{eq:generators101QD-2}
\end{align}
Each of four stabilizer elements in Eq. (\ref{eq:generators101QD-2}) represents an equivalence class of 16 degenerate elements.Thus the set of nine stabilizer generators for the $[[10,1]]_{QD}$ code is given by Eq. (\ref{eq:generators101QD-1}) and selecting any one element from each of the equivalence classes $S_6$ through $S_9$ in Eq. (\ref{eq:generators101QD-2}). 

In the DQ case, the resultant concatenated code has two blocks $(1a,1b,1c,1d,1e)$ and $(2a,2b,2c,2d,2e)$. We assign a copy of stabilizer generators of the inner QECC (\ref{eq:generators51}) to each block, yielding the eight of nine stabilizer generators being expressed as
\begin{align}
    &S_{1}=X^{1a}Z^{1b}Z^{1c}X^{1d},\nonumber\\
    &S_{2}=X^{1b}Z^{1c}Z^{1d}X^{1e},\nonumber\\
    &S_{3}=X^{1a}X^{1c}Z^{1d}Z^{1e},\nonumber\\
    &S_{4}=Z^{1a}X^{1b}X^{1d}Z^{1e},\nonumber\\
    &S_{5}=X^{2a}Z^{2b}Z^{2c}X^{2d},\nonumber\\
    &S_{6}=X^{2b}Z^{2c}Z^{2d}X^{2e},\nonumber\\
    &S_{7}=X^{2a}X^{2c}Z^{2d}Z^{2e},\nonumber\\
    &S_{8}=Z^{2a}X^{2b}X^{2d}Z^{2e}.
    \label{eq:generators101DQ-1}
\end{align}
The block label encoding of the generators of outer DFS (\ref{eq:generators21}) will give the remaining one generator 
\begin{align}
    &S_{9}=(X^{1a}X^{1b}X^{1c}X^{1d}X^{1e})(X^{2a}X^{2b}X^{2c}X^{2d}X^{2e}).
    \label{eq:generators101DQ-2}
\end{align}
Thus the set of stabilizer generators for the $[[10,1]]_{DQ}$ code is given by Eqs. (\ref{eq:generators101DQ-1}) and (\ref{eq:generators101DQ-2}).
We note that the generator in Eq. (\ref{eq:generators101DQ-2}) is passive and generates the two passively correctable errors $I^{\otimes10}$ and $X^{\otimes 10}$, in a manner similar to the case of the $[[6,1]]_{DQ}$ code [Eq. (\ref{eq:generators61DQ-2})].

\section{Conclusions and discussions}\label{sec-conclusions}

This work studied the codes obtained by concatenating quantum error-correcting codes and decoherence-free subspaces. This concatenation is suitable when both independent and correlated errors occur simultaneously. When the errors are sufficiently strongly correlated, the concatenation with the DFS as the inner code provides better entanglement fidelity, whereas for errors that are sufficiently independent, the concatenation with QECC as the inner code is preferable. The concatenation of QECC and DFS results in a kind of stabilizer structure that splits naturally into passive and active parts. Whereas the passive part functions like a DFS, the active one functions like a QECC.

As examples, we studied specifically the concatenation of a two-qubit DFS code with a three-qubit repetition code and five-qubit Knill-Laflamme code, under independent and correlated error models. In these examples, the considered QECCs are nondegenerate. A degenerate code such as Shor’s nine-qubit code can also be concatenated with a DFS. In this case, degeneracy in the actively correctable errors will have contributions not only from the DFS, but also from the QECC.

Whereas increasing the number of concatenated layers improves resistance to noise (with the asymptotic performance for the fault-tolerant systems obtained when the number of layers approaches infinity \cite{rahn2002exact}), the practical implementation also gets harder. Moreover, the pseudothreshold [defined as the maximum physical error rate $p$ below which $p_F(p) < p$] is insensitive to concatenation. In Fig. \ref{fig: fp6DQthresh}, the probability of failure ($p_F$) is plotted as a function of the probability of the individual qubit error ($p$) with concatenation depth ${L=1,2,3,4}$ for the $[[6,1]]_{DQ}$ code.

\begin{figure}[htp]
    \centering
    \includegraphics[width=8cm]{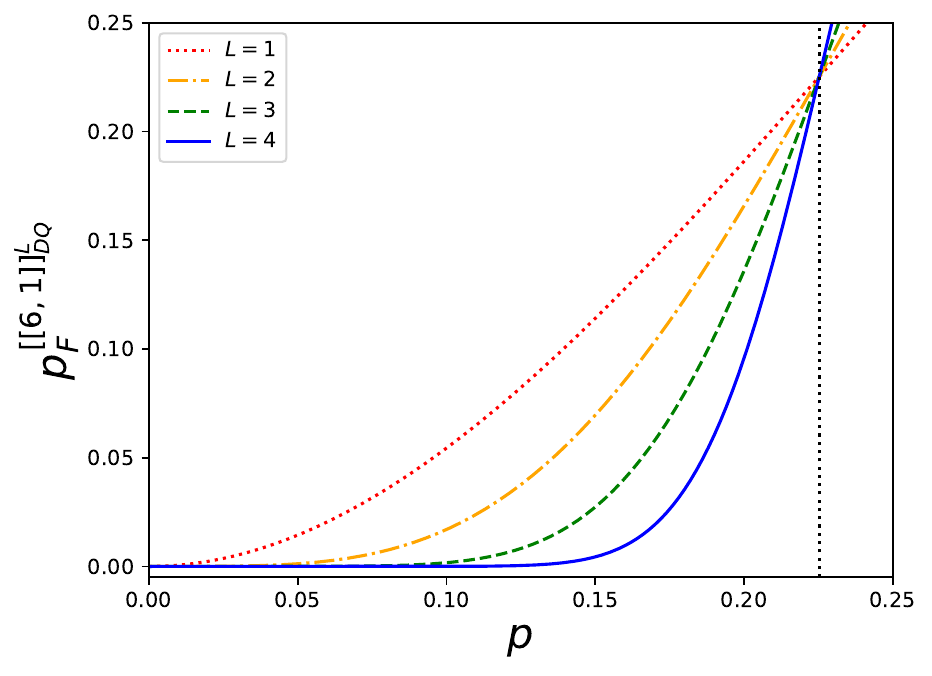}
    \caption{(Color online) Failure probability of the concatenated code $[[6,1]]_{DQ}$ with pseudothreshold probability $p_{thres}\approx 0.225$ for $L=1$, $2$, $3$, and $4$  concatenations, depicted as the red dotted, orange dash-dotted, green dashed, and blue solid plots, respectively.}
     \label{fig: fp6DQthresh}
\end{figure}

The pseudothreshold remains invariant, because when the innermost code fails, it precipitates the failure of all outer layers. The pattern is the same for the $[[6,1]]_{QD}$ and $[[10,1]]_{QD / DQ}$ codes, except that the pseudothreshold is slightly different. For independent errors, typically, the inner code with the higher pseudothreshold provides better logical noise suppression \cite{chamberland2016thresholds}. The performance parameters for the four codes discussed in this work, along with their pseudothresholds, are summarized in Table \ref{table:1}. In conclusion, the preferred order of concatenation will depend on both the level of correlation in the noise and the desired pseudothreshold. Here the concept of the pseudothreshold is the same as that of the threshold mentioned in Ref. \cite{rahn2002exact}. Note that the pseudothreshold should be distinguished from the asymptotic threshold relevant to realistic simulations of fault-tolerant computing \cite{svore2006flowmap}. Here different components of the circuit are allowed to fail at differing rates, so that the thresholds at different levels of concatenation no longer match, and an asymptotic analysis would be needed to determine the tolerable error rate $p$. 

\begin{table}[h!]
\caption{Error type, Hamming efficiency, modified Hamming efficiency, and pseudothreshold probability of the concatenated codes.}
\centering
\begin{tabular}{c c c c c} 
 \hline\hline
\backslashbox{Parameter}{Code} & $[[6,1]]_{QD}$ & $[[6,1]]_{DQ}$ & $[[10,1]]_{QD}$ & $[[10,1]]_{DQ}$ \\ [0.5ex] 
 \hline
 $E_{type}$ & $X,XX$ & $X,XX$ & $X,Y,Z,XX$ &  $X,Y,Z,XX$\\ 
 $\varphi$ & $1$ & $1$ & {$1$} & $1$ \\
 $\varphi^{\prime}$ & $0.4$ & $0.8$ & $0.444$ & $0.888$ \\
 $p_{thres}$ & $0.1293$ & $0.2252$ & $0.0298$ & $0.0579$ \\ [1ex] 
 \hline\hline
\end{tabular}
\label{table:1}
\end{table}
Given the presence of correlated noise, here the main task involves identifying the DFS and the selection of the compatible QECC in a way that facilitates implementing fault-tolerant quantum computation with the concatenated code in a given quantum processor. For example, in a two-dimensional decoherence-free photonic memory subspace protected from collective errors, if the independent errors also occur with nonvanishing probability, then a QECC can be encoded in it to improve protection. Our work motivates the study of other types of compatible QECC and DFS code concatenation. 

The QD and DQ codes can be considered as generalizing stabilizer codes to a new class of hybrid codes, which correct errors both actively and passively. The QECCs and DFSs form special cases of such a hybrid code.

\acknowledgements
N.R.D. and S.D. acknowledge financial support from the Department of Science and Technology, Ministry of Science and Technology, India, through the INSPIRE fellowship  
and the University Grants Commission India through the NET fellowship, respectively. N.R.D. also thanks Vinod Rao for insightful discussions during the early stages of this project. R.S. acknowledges partial support from the Science and Engineering Research Board through Grant No. CRG/2022/008345. N.R.D. is sincerely grateful to the Poornaprajna Institute of Scientific Research for its hospitality and conducive environment during his academic visit.

\bibliographystyle{apsrev4-1}
\bibliography{reference}

\end{document}